\documentclass[runningheads]{llncs}
\usepackage[utf8]{inputenc}
\usepackage[english]{babel}
\usepackage{amsmath,amssymb}
\usepackage{relsize}
\usepackage{graphicx}
\usepackage{float}
\usepackage{xcolor}
\usepackage{mathrsfs}
\usepackage{enumitem}
\usepackage{mathdots}
\usepackage{subcaption}
\usepackage{stmaryrd}
\usepackage{tikz}
\usepackage{tikz-cd}
\tikzcdset{scale cd/.style={every label/.append style={scale=#1},
    cells={nodes={scale=#1}}}}
\usepackage{hyperref}
\usetikzlibrary{babel}
\usetikzlibrary{arrows.meta}
\usetikzlibrary{decorations.pathmorphing}

\DeclareMathOperator{\id}{id}
\DeclareMathOperator{\Hom}{Hom}

\DeclareMathOperator{\inl}{inl}
\DeclareMathOperator{\inr}{inr}
\DeclareMathOperator{\unf}{\mathsf{u}}
\DeclareMathOperator{\Prop}{Prop}
\DeclareMathOperator{\Free}{Free}

\newcommand{\C}{\mathcal{C}}
\newcommand{\D}{\mathcal{D}}
\newcommand{\A}{\mathbb{A}}
\newcommand{\B}{\mathbb{B}}
\newcommand{\pow}{\mathscr{P}}
\newcommand{\Poset}{\mathbf{Poset}}
\newcommand{\Sets}{\mathbf{Sets}}
\newcommand{\BA}{\mathbf{BA}}
\newcommand{\DLatt}{\mathbf{DL}}

\newcommand{\SCL}{\mathbf{SCL}}
\newcommand{\uf}{\mathbf{uf}}

\newcommand{\Ord}{\mathrm{Ord}}

\newcommand{\lbk}{\llbracket}
\newcommand{\rbk}{\rrbracket}

\newcommand{\lds}{\llparenthesis}
\newcommand{\rds}{\rrparenthesis}

\newcommand{\smap}{\lbk-\rbk}
\newcommand{\sem}[1]{\lbk#1\rbk}

\newcommand{\op}{\text{op}}
\newcommand{\heart}{\heartsuit}
\newcommand{\Op}{\Xi}
\newcommand{\Opinit}{\mathrm{m}}

\newcommand{\trl}{trl}

\newcommand{\fold}[1]{\ulcorner #1 \urcorner}
\newcommand{\initsem}[1]{\|#1\|}

\newtheorem{assumption}{Assumption}

\title{A Categorical Approach to \\Coalgebraic Fixpoint Logic\thanks{%
This research is partially supported by the Leverhulme Trust Research Project Grant RPG-2020-232 and NWO grant No.~OCENW.M20.053.}}
\author{Ezra Schoen\inst{1} \and Clemens Kupke\inst{1} \and Jurriaan Rot\inst{2} \and Ruben Turkenburg\inst{2}}
\authorrunning{E.~Schoen et al.}
\institute{University of Strathclyde \and Radboud Universiteit Nijmegen}

\begin{document}

\maketitle
\begin{abstract}
    We define a framework for incorporating alternation-free fixpoint logics into the dual-adjunction setup for coalgebraic modal logics. We achieve this by using order-enriched categories. We give a least-solution semantics as well as an initial algebra semantics, and prove they are equivalent. We also show how to place the 
    alternation-free coalgebraic $\mu$-calculus in this framework, as well as PDL and a logic with a probabilistic dynamic modality.
\end{abstract}
\section{Introduction}

   Coalgebra provides a versatile framework for representing different types of state-based dynamic systems
   in a uniform way~\cite{rutten2003univcoalg}. At the heart of the coalgebraic theory lie the semantic notions 
   of behaviour and behavioural equivalence. It is well-known that modal logics provide the appropriate
   syntactic tools to specify and reason about labelled transition systems in a fully abstract way, i.e., in a way that 
   ensures that the language precisely characterises behavioural equivalence~\cite{HMTheorem1980}. 
   This is the basis for research into
   coalgebraic modal logics, i.e., modal logics that are developed parametric in the type of the coalgebra
   that the logic is supposed to be interpreted on. Many different modal logics, such as monotone modal logic, graded 
   modal logic and various probabilistic modal logics, have been
   shown to be instances of coalgebraic modal logics \cite{SchroderP09,CirsteaKPSV08,KupkeP11}. By placing those logics in a common framework, one is able to provide 
   generic proofs of expressivity, soundness and completeness of the logics that can then be instantiated to each of the logics, 
   thus avoiding the need for proving those results for each logic individually.
   Mathematically, the close connection between coalgebras and their corresponding coalgebraic modal logics has been represented in 
   the elegant framework of a dual adjunction that links $B$-coalgebras over a category $\C$ 
   to $L$-algebras over a category $\D$~\cite{KupkeKP04,klin2007beyond,PavlovicMW06}. This category $\D$ should be thought of as the category of algebras for the propositional logic of predicates, 
   whereas $L$ encodes the modal operators of the logic. The ``one-step'' semantics of the modal operators is then provided by a certain type of natural transformation connecting $L$ and $B$ across the adjunction.

   The biggest stumbling block for studying coalgebraic modal logics abstractly is probably the inherent ``one-stepness'' of the theory: 
   the transition structure typically only allows to look one step ahead in the model.
   Consequently, modal operators are usually one-step and axioms of the logics are assumed 
   to be non-iterative~\cite{forsterschroeder2020} i.e., not allowing nesting of modal operators (exceptions such as~\cite{dahlqvistcmcs2016} confirm
   the rule).
   Fixpoint operators pose a problem in this context as they specify properties 
   that can look {\em arbitrarily deep} into the model.
   There are several existing approaches to adding fixpoint operators to coalgebraic modal logics 
   \cite{Venema06,cirstea2011tableaux,schroder2010flat}
   and even a coalgebraic model-checking tool for those logics~\cite{HausmannHPSS24}.
   However, none of the existing approaches to fixpoint logics 
   provide a {\em categorical} treatment of the 
   fixpoints within the above outlined 
   dual adjunction framework. This means that existing coalgebraic fixpoint logics are developed
   on the category of sets and that results such as 
   invariance under behavioural equivalence cannot be proven
   in an abstract, diagrammatic way. 

   The main contribution of this paper is to extend the dual adjunction framework to incorporate alter\-nation-free 
   fixpoint logics.  We will first introduce the key concept of an {\em unfolding system}
   that contains as ingredients a one-step logic, 
   a functor corresponding to the fixpoint operators and a natural transformation
   that we call unfolding and that is used to represent the unfolding of fixpoints. 
   In order to ensure existence of fixpoints we will assume that both the category $\D$
   and the functors on $\D$ corresponding to the logic are $\Poset$-enriched. 
   We will define the semantics of a given unfolding system as the least solution of an 
   unfolding operation.  
   After we present the abstract framework for fixpoint logics in Section~\ref{sec:semantics} 
   we will demonstrate
   how to place several examples in the framework: a positive modal 
   logic with a transitive closure modality, a probabilistic version of a similar
   logic whose fixpoint operator resembles iteration 
   in PPDL~\cite{kozen1985ppdl,gu2020ppdl} and finally the positive fragment of
   PDL~\cite{rot2021steps} without tests. 
   We will then prove key technical results: a diagrammatic
   proof of ``adequacy'' of the given fixpoint logic and, in Section~\ref{sec:initsem}, the fact 
   that the semantics can equivalently be obtained as more standard initial algebra semantics.

   In the final part of the paper, consisting of Section~\ref{sec:afcfl} 
   and Section~\ref{sec:negations}, we show how to place the alternation-free 
   fragment of the coalgebraic $\mu$-calculus (in the sense of~\cite{cirstea2011tableaux}) 
   into our framework. This is
   done in two stages: first, in Section~\ref{sec:afcfl}, we translate
   the positive fragment of the coalgebraic $\mu$-calculus into 
   a syntax that allows us to place the logic into our framework by 
   representing the logic and its semantics as an unfolding system. 
   In Section~\ref{sec:negations} we then provide a general recipe 
   for adding negations to a fixpoint logic represented in our framework.
   The latter will in particular show that the full alternation-free
   fragment of the coalgebraic $\mu$-calculus can be represented 
   using a suitable functor $L$ and the associated 
   initial $L$-algebra semantics. 
   Finally, in Section~\ref{sec:conclusion}, 
   we conclude with some ideas for ongoing and future work.

\section{Preliminaries}

We will assume familiarity with basic category theory, coalgebra and 
modal logic. 

\subsubsection{Notation}
    Throughout the paper we will use $\pow$ to denote the covariant powerset functor. 
    The functor $P$ denotes the left adjoint of the dual adjunction we will be working with. 
    In many concrete
    instances of this adjunction, $P$ will thus denote the contravariant power set functor.

\subsubsection{Fixpoints}\label{sec:fix}

We will heavily rely on Kleene's fixpoint theorem and its generalisation by Cousot \& Cousot~\cite{cousot1979fixpoint}  
that states that the least fixpoint of a monotone function $f: L \to L$ on a complete lattice $L$ exists and can be obtained as the limit of the sequence
$f^0 = \perp$, $f^{i+1} = f(f^i)$ for an arbitrary ordinal $i$ and
$f^j = \bigvee_{i<j} f^i$ for limit ordinals $j$. We will not make explicit use of the dual statement concerning greatest fixpoints.

\subsection{One-step Logics}\label{sec:onestep}

We model coalgebraic modal logics via the dual adjunction approach, cf.~e.g.~\cite{klin2007beyond}.
Throughout the paper we will assume to work in a setting where we
are given:
\begin{itemize}
    \item A category $\C$ of \emph{spaces}, which are the carriers for coalgebras;
    \item A category $\D$ of \emph{algebras} for some underlying `propositional' logic; 
    \item A dual adjunction $P:\C\to \D^\op$ and $Q:\D^\op\to \C$ 
    with $P\dashv Q$, i.e., 
    for all $X \in \C$, $A \in \D$ we have 
    $
    \Hom_\C(X, QA)\cong \Hom_\D(A, PX)
    $
    and this isomorphism is natural in both $X$ and $A$;
    \item and an endofunctor $B:\C\to \C$ specifying
    the coalgebra type.
\end{itemize}

The adjunctions $P\dashv Q$ we consider are \emph{logical connections}, as in \cite{kurz2013connect}. That is to say, $P$ and $Q$ are both of the form $\Hom(-, \Omega)$, where $\Omega$ is a so-called `ambimorphic' object, living in both $\C$ and $\D$. 

\begin{definition}\label{def:onesteplogic}
    Given a functor $B:\C\to \C$, a pair $(L_0,\delta)$ consisting of
    a functor $L_0:\D\to\D$ and a natural transformation
    $\delta:L_0P\to PB$
    is called a {\em one-step logic for $B$}.  
\end{definition}
The natural transformation captures the ``one-step semantics'' of the logic. 
Crucially, while we will be able to place fixpoint logics into the dual adjunction
framework, we will see that those logics cannot be described as one-step logics. This reflects
the fact that fixpoint logics are inherently multi-step, as fixpoint operators can be 
used to inspect the model arbitrarily deep.

\subsubsection{Examples of One-Step Logics}

\begin{example}\label{ex:positivemodal}
We take $\C$ to be the category $\Sets$, and $\D$ to be the category $\DLatt$ of 
distributive lattices. 
One half of the adjunction is formed by the powerset functor $P:\Sets\to\DLatt^\op$,
which can be seen as exponentiation $2^-$. 
Similarly, for $Q: \DLatt^\op \to \Sets$ we use $\Hom_\DLatt(-,2)$
. It is well known that $P\dashv Q$, forming the logical connection using $2$ as the ambimorphic object. 

For our behavior functor, we fix a set $\Prop$ of propositional letters, and define $B:\Sets\to\Sets$ as
$
BX = \pow\Prop \times \pow X
$.
For our one-step logic, we let $L_0:\DLatt\to\DLatt$ be given by
$
L_0\A = \Free\left(\{\Diamond a\mid a\in A\}\cup \{p\mid p\in \Prop\}\right)/{\approx}
$
where $\Free$ generates the free distributive lattice on a set of 
generators,
and $\approx$ is the least congruence satisfying
$\Diamond a \approx (\Diamond a \wedge\Diamond b)$ whenever $a \leq b$.

We note for future reference that quotienting out $\approx$ exactly ensures that $\Diamond a \leq \Diamond b$ whenever $a \leq b$. We obtain a one-step semantics $\delta:L_0P\to PB$ via
\[
\delta:\begin{cases}p&\mapsto \{\langle m, U\rangle \mid p\in m\}\\\Diamond v&\mapsto \{\langle m, U\rangle \mid v\cap u\neq\varnothing\}\end{cases}
\]
and extending freely. This yields 
the expected semantics of (positive) modal logic as follows: consider a 
$B$-coalgebra $(X,\gamma)$, let $(\Psi,\alpha)$ be the initial $L_0$-algebra
 (to be thought of as algebra of formulas) and consider the $L_0$-algebra 
 $P\gamma \circ \delta: L_0  PX \to PX$. The initial algebra map $\smap$ will satisfy the following
 \begin{eqnarray*}
    \sem{\alpha ( \Diamond a)} & = & P\gamma ( \delta_X (L \smap(\Diamond a))) 
     =  \{ x \in X \mid \gamma(x) \in \delta_X(\Diamond \sem{a}) \} \\
     & = & \{ x \in X \mid \gamma(x) \in \{\langle m,U
     \rangle \in BX \mid U \cap \sem{a} \not= \emptyset \} \}
 \end{eqnarray*}
 which expresses precisely that $\Diamond a$ will be true at those states that have at least
 one successor ``satisfying'' $a$. 
\end{example}

\begin{example}\label{ex:quantmodal}
For $\C$ we again take the category of sets. For $\D$, we take distributive lattices that come equipped with subconvex combinations. By this, we mean that if $\lambda_1, \dots,\lambda_n\in [0,1]$ with $\sum_{i}\lambda_i < 1$, then for each $a_1,\dots, a_n\in \A$, we obtain an element
\[
\sum_i\lambda_ia_i = a\in \A
\]
Moreover, we demand that
\[
1\cdot a = a,\quad \sum_i\lambda_i(\sum_{j}\mu_{i,j}a_{i,j}) = \sum_{i,j}(\lambda_i\cdot \mu_{i,j})a_{i,j}
\]
and
\[
\lambda a + \mu (b\vee c) = (\lambda a + \mu b)\vee(\lambda a + \mu b)
\]
and similar for $\wedge$. Let $\SCL$ be the category with objects subconvex lattices, and morphisms the maps preserving both subconvex and lattice structure.

A key example of a subconvex lattice is given by $[0,1]$, with $\max$ and $\min$ as lattice operations, and subconvex structure given in the obvious way. Note also that if $\A$ is a subconvex lattice, then so is $\A^X$ with pointwise structure; hence, we obtain our main examples as $[0,1]^X$, with $X$ any set.

Note also that if $f:X\to Y$ is any function, then $f^*:[0,1]^Y\to [0,1]^X$ is a subconvex lattice homomorphism. Hence, we obtain a functor $P:\Sets\to \SCL^\op$ given by $X\mapsto [0,1]^X$. Vice versa, we clearly have a morphism $Q:\SCL^\op\to\Sets$ given by $\A\mapsto \Hom_{\SCL}(\A,[0,1])$. It is easy to see that there is an adjunction $P\dashv Q$, since they are both of the form $\Hom(-, [0,1])$. 
Finally, for our fixpoint extensions later on, 
note that $\SCL$ is enriched in posets in the obvious way, 
and that $PX$ is a complete lattice for all $X$. 

For a set $X$, let $\Delta(X)$ be the set of finitely supported subdistributions on $X$:
\[\Delta(X) \mathrel{:=} \{\mu:X\to [0,1]\mid \mu(x) = 0\text{ all but finitely many }x, \sum_{x\in X}\mu(x) \leq 1\}.\]
$\Delta$ is an endofunctor 
on $\Sets$, where $\Delta f$ maps a distribution $\mu$ on $X$ to $\mu_f$ with 
$
\mu_f(y) = \sum_{f(x) = y}\mu(x).
$
Now fix a set $A$ of labels, and let $B$ be the functor $BX = [0,1]^A\times \Delta(X)$. We can think of a $B$-coalgebra as a probabilistic 1-player game, where in a given state, the player may select a label $A$ to obtain a `payout', or probabilistically transition to a next state; but if the player takes the probabilistic transition, there is a possibility of failure, since the probabilities need not add up to 1. 

In this context, we may posit the following one-step logic $L_0:\SCL\to \SCL$: for a given subconvex lattice $\A$, we let $L_0\A$ be the free subconvex lattice generated by $\{\Diamond a\mid a\in \A\}\cup \{p\mid p\in A\}$, quotiented by the equations
\begin{align*}
\sum_i \lambda_i \Diamond a_i &\approx \Diamond(\sum_i\lambda_i a_i)\\
\Diamond a \wedge \Diamond b &\approx \Diamond a&\text{ whenever }a\wedge b = a
\end{align*}

Intuitively, $\Diamond a$ should be read as `the expected value of $a$'; this is why we demand that $\Diamond$ acts linearly and preserves the order, but does \emph{not} necessarily preserve lattice structure (as $\mathbb{E}[\max(X,Y)]$ is usually strictly greater than both $\mathbb{E}X$ and $\mathbb{E}Y$). 
Using the intuition of `expectation', we obtain a one-step semantics $\delta:L_0P\to PB$ by
\begin{align*}
\delta(\Diamond u)&:\langle\pi,\mu\rangle \mapsto \sum_{x\in X}\mu(x)u(x) =: \mathbb{E}_\mu(u)\\
\delta(p)&:\langle\pi,\mu\rangle\mapsto \pi(p)
\end{align*}
and extending freely; it is easy to verify that $\delta$ respects the equations for $L_0$, so this indeed is a well-defined subconvex lattice morphism $L_0P\to PB$. Naturality is also easy to verify. 
\end{example}

\subsection{Enriched Categories}

We review the concepts from enriched categories that we use in this paper. Since we only consider categories enriched in $\Poset$, some things simplify. In particular, for our purposes we do not need the general case of weighted (co)limits, which allows us to stick close to unenriched category theory. For a more in-depth 
treatment of enriched categories see~\cite{kelly2005enriched}.

\begin{definition}
A \emph{$\Poset$-enriched category} $\C$ is a category $\C$, together with a partial order ${\leq} = {\leq_{A,B}}$ on each homset $\Hom_\C(A,B)$, such that
\[
-\circ-:\Hom_\C(B,C)\times\Hom_\C(A,B)\to \Hom_\C(A,C)
\]
is an order-preserving map for all $A,B,C$. 
\end{definition}
A key example of a $\Poset$-enriched category is $\Poset$ itself: one can order morphisms `pointwise', that is,
$
f\leq_{A,B} g \text{ iff }f(x)\leq_B g(x)\text{ for all }x\in A
$.
In fact, all examples of $\Poset$-enriched categories in this paper are ordered pointwise in a similar way. That is, we consider categories $\D$ where objects are ordered structures, and the morphisms respect the orders; we then enrich $\D$ in $\Poset$ by ordering the morphisms pointwise. 
There is also a notion of enriched functor.
\begin{definition}
Let $\C,\D$ be $\Poset$-enriched categories. A \emph{$\Poset$-enriched functor} is a functor $F:\C\to\D$ such that for all parallel arrows $f,g:A\to B$ in $\C$, we have
$
f\leq g \implies Ff\leq Fg
$.
\end{definition}
Finally, we have a notion of enriched coproduct:
\begin{definition}
Let $\C$ be a $\Poset$-enriched category, and let $A,B$ be objects in $\C$. We get a functor $F:\C\to\Poset$ given by
$
FX = \Hom_\C(A,X)\times\Hom_\C(B,X)
$.
We call an object $C$ of $\C$ an \emph{enriched coproduct of $A$ and $B$} if there is a natural isomorphism of functors
$
F\cong\Hom_\C(C,-)
$.
\end{definition}
That is to say, an enriched coproduct of $A$ and $B$ is a coproduct $A+B$ such that the (natural) isomorphism
$
\Hom_\C(A,X)\times\Hom_\C(B,X) \cong \Hom_\C(A+B,X)
$
is an isomorphism of \emph{posets}, not merely sets. 

\section{Semantics of Fixpoint Logics}~\label{sec:semantics}

\subsection{Unfolding Systems}

To define a fixpoint logic, we take a `two-tiered' 
approach: one starts with a `base' modal (one-step) logic, 
to which one adds fixpoint modalities. Hence 
we will work in a setting where we are given  
a one-step logic as in Definition~\ref{def:onesteplogic},
together with an enrichment that supports fixpoint operators.

\subsubsection{Enriching the Dual Adjunction}

To define the semantics of a fixpoint formula as a `least solution', 
it is necessary to be able to compare predicates (i.e., elements of objects in $\D$). 
Since we intend a fully diagrammatic approach, we 
generalize comparing elements to comparing \emph{maps} $X\to Y$ in $\D$. 
This is a generalization indeed: if $\D$ is a concrete 
category where each object is an ordered set, then maps can be ordered pointwise. 
To be able to compare maps, we fix an enrichment of $\D$ in $\Poset$. We need some additional 
assumptions:
\begin{assumption}\label{ass:P}
\begin{enumerate}[label = (\roman*)]
\item \label{ass:P1} We assume that $\Hom_\D(X, PY)$ is a complete lattice for each $Y\in \C$, 
and that for each morphism $f:Y'\to Y$, the map $Pf\circ -:\Hom_\D(X,PY)\to \Hom_\D(X, PY')$ 
preserves the lattice structure and arbitrary directed joins. 

\item \label{ass:P2} Since $P$ is a left adjoint, it preserves colimits; however, we require the stronger condition that $P$ maps coproducts to \emph{enriched} coproducts.
\end{enumerate}
\end{assumption}
\begin{remark}
A natural candidate for $\D$ is the category $\BA$ of Boolean algebras, as we have the powerset-ultrafilter adjunction $P\dashv \uf$ between $\Sets$ and $\BA^\op$. Note, however, that if we enrich $\BA$ by pointwise order on the morphisms, the morphisms end up being ordered discretely. For, if $f\leq g$ in $\Hom_\BA(\A,\B)$, then for each $x\in \A$ we have $f(\neg x)\leq g(\neg x)$. Hence,
$
g(x) = \neg g(\neg x) \leq \neg f(\neg x) = f(x)
$, showing $g\leq f$ as well. 

It may be possible to enrich $\BA$ via a different, non-pointwise strategy; however, it can be shown that there is no enrichment of $\BA$ that makes $\Hom(X,PY)$ a complete lattice for all $X, Y$; hence, we will work with the category $\DLatt$ of distributive lattices, rather than $\BA$. 
This could be seen as an analogue of the constraint on 
fixpoint equations, that fixpoint variables only occur \emph{positively}. 
In section~\ref{sec:negations}, we outline how negations can be added 
into the picture outside of the fixpoint equations. 
\end{remark}

\subsubsection{Logical Functors}
In order to be able to enrich a given logic with fixpoints we will
assume to be given a one-step logic $(L_0, \delta)$ for $B$ 
 and posit a functor 
$L:\D\to \D$ representing the fixpoint modalities. 
We assume moreover that both $L_0$ and $L$ are $\Poset$-enriched functors, 
and that $L$ has an initial algebra.
Throughout the paper we will write $\Phi$ for the initial $L$-algebra, 
with structure map $\alpha:L\Phi\to \Phi$ (note that $\alpha$ is an isomorphism); 
it should be thought of as the `algebra of formulas'. 

\begin{remark}
    Intuitively the assumption that the functor $L_0$ is enriched means that all modal 
    operators in the base logic are monotone.   
\end{remark}

We will now define a (alternation-free) fixpoint logic 
as an extension of a basic one-step logic. 
A key element is the so-called unfolding operation for the fixpoints. 

\begin{definition}
    Let $\D$ be a $\Poset$-enriched category. 
    Let $B:\C \to \C$ be a functor, and let $L_0:\D \to \D$ be a $\Poset$-enriched functor. 
    An {\em unfolding system} for $L_0$ is a pair 
    $(\unf:L\to \id + L_0L,\delta:L_0 P \to P B)$ where 
    $L:\D \to \D$ is a $\Poset$-enriched functor that has an initial algebra $\Phi$,  
    $\unf$ is a natural transformation and $(L_0, \delta)$ is a one-step logic for $B$ in the sense of Def.~\ref{def:onesteplogic}.
    %
\end{definition}

The natural transformation $\unf:L\to \id + L_0L$ provides the key to define the semantics of fixpoint modalities. 

Intuitively, if we think of an element $\ell$ of $L\A$ as being a loop of some type, we see that there are two `branches' to 
unfolding $\ell$; we may exit the loop, 
which yields an element of $\A$, or we may take a step in the model (represented by a modality taken from $L_0$), after which we continue with a loop. 

In more detail, a generic element of $L\A$ may look like 
$\ell(x_1,\dots, x_n)$, with $\ell$ 
some fixpoint modality, and $\bar x = (x_1,\dots, x_n)\in \A^n$. 
We want $\ell$ to satisfy a fixpoint equation
\[
\ell(\bar x) = \sigma(\bar x, g_1\cdot \ell_1(\bar x),\dots, g_k\cdot \ell_k(\bar x))
\]
Here $\sigma$ is some `propositional' 
expression, and each $\ell_j$ occurs \emph{guarded} 
by a one-step modality $g_i$. The transformation $\unf$ 
replaces the LHS of each such expression with the RHS; 
and the RHS lives in $\A + L_0L\A$. 

\subsection{Definition of the Semantics}

We are now ready to define the semantic map $\smap$. 
\begin{definition}\label{def:smap}
Let $(\unf,\delta)$ be an unfolding system for $B$ and let $(X,\gamma)$ be a 
$B$-coalgebra. The semantic map $\smap_\gamma:\Phi \to PX$ on $(X,\gamma)$
 is defined as the least map $t:\Phi \to PX$ making the following diagram commute:
\[
\begin{tikzcd}
\Phi\arrow[d, "\alpha^{-1}"]\arrow[r, "t"] & PX\\
L\Phi\arrow[d, "\unf"] & PX + PBX\arrow[u, "{[PX, P\gamma]}"]\\
\Phi +  L_0L\Phi\arrow[d, "{\id +L_0\alpha}"] \\
\Phi + L_0\Phi \arrow[r, "t + L_0t"] & PX + L_0PX\arrow[uu, "\delta"]
\end{tikzcd}
\]
When the $X$ is clear from the context we often will drop the subscript and simply denote the
semantic map by $\smap$.
\end{definition}
Intuitively, the left-hand side of the above diagram takes a formula, and unfolds the top-level fixpoint modalities, to obtain a one-step 
$L_0$-formula over $\Phi$. Using the semantic map, this can be interpreted 
as a one-step $L_0$-formula over $PX$, which we can reduce to a predicate in $PX$ 
using the one-step semantics $\delta$, and the coalgebra structure $\gamma$. 

In order to see that $\smap$ is well-defined, we note that the above formulation
is equivalent to saying that $\smap$ is the least fixpoint of the operator
\begin{equation}\label{equ:Op}
\begin{array}{rcl}
   \Op_\gamma: \Hom_\D(X,PX) & \to & \Hom_\D(X,PX) \\ 
   \\ 
        t & \mapsto & [PX, P\gamma] \circ \delta \circ (t + L_0t) \circ (\id + L_0\alpha)
        \circ \unf \circ \alpha^{-1}
\end{array}
\end{equation}
That $\Op_\gamma$ is a monotone operator can be seen as follows:
\begin{enumerate}
	\item as $\D$ is order-enriched, both pre- and post-composition of morphisms are monotone, and 
	\item the operation that maps $t$ to $t + L_0 t$ is monotone.  
\end{enumerate}
The first fact is immediate from the definition of enriched categories. The second fact follows from the operation $t \mapsto L_0 t$ being monotone ($L_0$ is an enriched functor) and because the cotupling 
operation is monotone as we have
\[ \Hom(X + Y,Z) \cong \Hom(X,Z) \times \Hom(Y,Z) \]
by definition of the coproduct and because the above isomorphism is an isomorphism of posets.  
As $\Op$ is monotone and as $\Hom_\D(X,PX)$ is a complete lattice by assumption, the
least fixpoint $\smap$ of $\Op$ exists.

\begin{remark}
While the definition of $\lbk-\rbk$ is in terms of a least fixpoint, the same setup can be used to define greatest fixpoints as well, by inverting the order on the homsets. However, we are as yet restricted to only a single \emph{type} of fixpoint. 
\end{remark}

\subsection{Examples} 

\begin{example}\label{ex:diamondstar}
As a first example, we give a simple logic of transitive closure in Kripke frames. 

We start with the one-step logic as given in example \ref{ex:positivemodal}. To obtain an unfolding system, we take $L:\DLatt\to\DLatt$ similarly to $L_0$, but using $\Diamond^*$ rather than $\Diamond$:
\[
L\A = \Free\left(\{\Diamond^* a\mid a\in A\}\cup \{p\mid p\in \Prop\}\right)/{\approx}
\]
where $\approx$ is as in example~\ref{ex:positivemodal}. We get an unfolding transformation $\unf:L\to \id + L_0L$ by
\[
\unf:\begin{cases}p&\mapsto \inr(p)\\\Diamond^*a&\mapsto \inl(a)\vee \inr(\Diamond\Diamond^*a)\end{cases}
\]
To spell out the concrete description of the above logic, we have formulas defined via
\[
\phi,\psi ::= p\in \Prop\mid \top \mid \bot \mid \phi\wedge\psi \mid \phi\vee\psi \mid \Diamond^*\phi
\]
For a given coalgebra seen as a Kripke model $\mathfrak{M} = (W, R, V)$, consider the usual definition of satisfaction of formulas, given by
\begin{align*}
x\Vdash p &\text{ iff } x \in V(p)\\
x\Vdash \Diamond^*\phi &\text{ iff }\exists v_1,\dots, v_n\text{ such that }xRv_1R\dots Rv_n\text{ and } v_n\Vdash \phi
\end{align*}
We claim that our approach yields the same semantics, in the sense that
\[
\lbk\phi\rbk = \{x\in W \mid x \Vdash \phi\}
\]
To see this, first note that the map $t:\Phi\to PW$ given by $\phi\mapsto \{x\in W\mid x\Vdash \phi\}$ is a solution to the diagram in definition~\ref{def:smap}; hence, $\lbk \phi\rbk \subseteq t(\phi)$, as $\lbk-\rbk$ is the least solution. For the other direction, let $t^{(k)}$ be the map that interprets $\Diamond^*$ as `reachability \emph{in at most $k$ steps}'; by induction on $k$, one can show that each $t^{(k)}$ is below all solutions to the diagram, and hence each $t^{(k)} \leq \smap$. Since $t$ is clearly the supremum of the $t^{(k)}$, we also have $t(\phi)\subseteq \lbk \phi\rbk$ for all formulas $\phi$. 
\end{example}

\begin{example}\label{ex:pdl}The above example may be extended to cover Propositional Dynamic Logic (PDL). We take the same base adjunction, but adjust $B$ and $L_0$ to include a set of modalities. That is, for a set $\Pi$ of `atomic programs', we let $BX = \pow\Prop \times (\pow X)^\Pi$, and take $L_0:\DLatt\to\DLatt$ to be
\[
L_0\A = \Free(\Prop \cup \{\langle \pi \rangle a\mid \pi \in \Pi, a\in \A\})/{\approx}
\]
with obvious action on morphisms. Then just as before, we obtain a one-step semantics as
\[
\delta:\begin{cases}p&\mapsto \{\langle m, u\rangle \mid p\in m\}\\\langle \pi\rangle v&\mapsto \{\langle m, u\rangle \mid v\cap u(\pi)\neq\varnothing\}\end{cases}
\]
For PDL, we define (composite) programs via the following grammar:
\[
\alpha ::= \pi\in \Pi \mid \epsilon \mid \alpha \cup \alpha \mid \alpha ; \alpha \mid \alpha^*
\]
We write $\bar \Pi$ for the set of programs. There is a function $g:\bar\Pi \to \bar \Pi$ such that $g(\alpha)$ is equivalent to $\alpha$, and $g(\alpha)$ 
is of the form $\sum_i \pi_i;\alpha_i$ or $(\sum_i\pi_i;\alpha_i) \cup \epsilon$. Existence of $g$
can be proven as the normal form in~\cite[Thm.~4.4]{brzozowski} using the
Brzozowski derivative.

We let $L:\DLatt\to\DLatt$ be given by 
\[
L\A := \Free(\Prop \cup \{\langle \alpha \rangle a\mid \alpha \in \bar\Pi, a\in \A\})/{\approx}
\]
with obvious action on morphisms. We define $\unf:L\to \id +L_0L$ as
\begin{align*}
\unf(p) &:= \inr(p)\\
\unf(\langle \alpha\rangle a) &:= \begin{cases} \inr\left(\bigvee_i \langle \pi_i\rangle \langle \alpha_i\rangle a\right) & g(\alpha) = \sum_i \pi_i; \alpha_i\\\inr\left(\bigvee_i \langle \pi_i\rangle \langle \alpha_i\rangle a\right)\vee \inl(a) & g(\alpha) = (\sum_i \pi_i;\alpha_i) \cup \epsilon\end{cases}
\end{align*}
Concretely, we see that $g(\pi^*) = \pi;\pi^* \cup \epsilon$, and hence (omitting coproduct inclusions) we have
\[
\unf(\langle \pi^*\rangle a) = \langle \pi \rangle \langle \pi^*\rangle a \vee a
\]
Just like $\Diamond^*$ above, it can be shown that $\lbk \langle\alpha\rangle\phi\rbk$ has the usual denotation.
\end{example}

\begin{example}\label{ex:quantex} As a third example, we give a quantitative logic for one-player games. 

In example~\ref{ex:quantmodal}, we gave a simple quantitative modal logic. We may define a functor $L$ by setting
\[
L\A := \Free(\{\sigma_qa, \Diamond^*a\mid a\in \A, q\in [0,1]\})
\]
(where $\Free$ denotes the `free subconvex lattice'-functor) and unfolding the generators as
\begin{align*}
\unf(\sigma_qa) &= qa + (1-q)\Diamond \sigma_qa\\
\unf(\Diamond^*a) &= a \vee \Diamond\Diamond^*a
\end{align*}
Under this unfolding, $\lbk\sigma_q\phi\rbk$ will assign to a state $x$ the expected outcome of a $q$-probabilistic strategy, where with probability $q$, the player chooses the payout corresponding to $\phi$, and with probability $(1-q)$ she decides to play on. 

On the other hand $\lbk \Diamond^*\phi\rbk$ assigns to a state $x$ the expected outcome of the \emph{optimal} strategy, where the player chooses to play on only if the outcome she can expect from playing on is greater than the payout she can get now.
\end{example}

\subsection{Invariance Under Behavioural Equivalence}\label{sec:invariance}

A fundamental property of coalgebraic modal logic is that the semantics of the logic
is invariant under behavioural equivalence. This follows from invariance under coalgebra morphisms
which can be concisely expressed as in~\eqref{eq:invariance} below. The proposition
shows that invariance under behavioural equivalence also holds in the presence of fixpoint
operators.
\begin{proposition}
    For all coalgebra morphisms $f:(X_1,\gamma_1) \to (X_2, \gamma_2)$
    we have
    \begin{equation}\label{eq:invariance}
        Pf \circ \smap_{\gamma_2} = \smap_{\gamma_1}
    \end{equation}
\end{proposition}
\begin{proof}
    Recall the defintion of $\Op$ from \eqref{equ:Op} on page~\pageref{equ:Op}.
    Using naturality of $\delta$ it is a matter of routine checking that
    $$\Op_{\gamma_1} = \Hom(-,Pf) \circ \Op_{\gamma_2}.$$
    As a consequence we have
    $Pf \circ \smap_{\gamma_2} \geq \smap_{\gamma_1}$
    because
    $$\Op_{\gamma_1} (Pf \circ \smap_{\gamma_2} ) = Pf \circ \Op_{\gamma_2} (\smap_{\gamma_2} ) =
    Pf \circ \smap_{\gamma_2},$$
    i.e., $Pf \circ \smap_{\gamma_2}$ is a fixpoint of $\Op_{\gamma_1}$ and $\smap_{\gamma_1}$ is the
    least such fixpoint.

    We will now show by ordinal induction that for all $i \in \mathrm{ORD}$ we have
    \[ \Op_{\gamma_1}^i \geq \Hom(-,Pf) \circ \Op_{\gamma_2}^i \]
    where $\Op_{\gamma_1}^i$ and $\Op_{\gamma_2}^i$ are the approximants defined in Section~\ref{sec:fix}.
    \begin{description}
        \item[Case] $i = j + 1$. Then
        \begin{eqnarray*}
            \Op_{\gamma_1}^i & = & \Op_{\gamma_1}(\Op_{\gamma_1}^j) 
                        \geq \Op_{\gamma_1}(\Hom(-,Pf)\circ \Op_{\gamma_2}^j) \\
                        & = & \Hom(-,Pf) \circ \Op_{\gamma_2}(\Op_{\gamma_2}^j)
                         =  \Hom(-,Pf) \circ \Op_{\gamma_2}^i
        \end{eqnarray*}
        \item[Case] $i$ is a limit ordinal. then
        \begin{eqnarray*}
            \Op_{\gamma_1}^i & = & \bigvee_{j < i} \Op_{\gamma_1}^j 
            \geq  \bigvee_{j < i} \Hom(-,Pf)\circ \Op_{\gamma_2}^j \\
            & \stackrel{\mbox{\tiny Ass.~\ref{ass:P}}}{=}&  \Hom(-,Pf)\circ \bigvee_{j < i} \Op_{\gamma_2}^j
             =   \Hom(-,Pf) \circ \Op_{\gamma_2}^i
        \end{eqnarray*}
    \end{description}
    The claim follows now as $\D$ is locally small which implies the hom sets
    $\Hom(\Phi, P X_i)$ for $i=1,2$ are sets, and thus $\smap_{\gamma_i}=\Op_{\gamma_i}^i$ for a
    suitably large ordinal $i$.
\end{proof}

\subsection{Guarded Equations of Arbitrary Depth}

The framework presented so far seems to place some restrictions on the shape of the fixpoint equations; consider the following toy logic:
\[
\phi ::= \top \mid \bot \mid \phi \vee \phi \mid \phi\wedge\phi \mid \Diamond^{2*}\phi
\]
with the intended interpretation of $\Diamond^{2*}$ in a transition system $(X,\to)$ being the least solution to
$
\lbk \Diamond^{2*}\phi \rbk = \lbk \phi\rbk \vee \{x\in X\mid \exists y,z:x \to y \to z\text{ and }z\in \lbk \Diamond^{2*}\phi\rbk\}
$
The natural guarding 1-step logic is positive modal logic; but this would yield an unfolding transformation
$
\unf:\Diamond^{2*}a \mapsto a\vee \Diamond\Diamond \Diamond^{2*}a
$
and the RHS is not of the form $\A + L_0L\A$. A possible remedy would be to add an extra operator $\Diamond^{2*+1}$ to the logic, which would allow the depth-one transformation
\begin{align*}
\Diamond^{2*}a\mapsto a \vee \Diamond \Diamond^{2*+1}a & \qquad &
\Diamond^{2*+1}a\mapsto \Diamond\Diamond^{2*}a
\end{align*}
but this requires an extension of the logic, and is somewhat ad-hoc.

A second issue is that in the current presentation, \emph{all} modalities are treated as fixpoint modalities. But most fixpoint logics contain a mixture of fixpoint and one-step modalities.
In fact both fixpoint modalities of higher depth and one-step modalities (which may be seen as `depth 0 fixpoint modalities') can be accomodated by our setup, in a uniform way.

\begin{proposition}\label{prop:depth}
Assume that $L+L_0$ generates an algebraically free monad $T$. Let $\unf :L \to \id + L_0T$ be a natural transformation. Then we obtain a natural transformation $\unf^*:T\to \id + L_0T$ such that $\unf$ factors through $\unf^*$.
\end{proposition}

Note that the assumption is satisfied whenever $L$ and $L_0$ are given via a finitary presentation, i.e. as a quotient of a free algebra on some finitary operations. The above proposition allows us to apply our results even in the less restrictive case where the codomain of $\unf$ is expanded to allow formulas of arbitrary depth, and including both one-step and fixpoint modalities. For instance, in the above toy logic, we get a transformation $\unf:L\to \id + L_0L_0L$ sending $\Diamond^{2*}a$ to $a \vee \Diamond\Diamond \Diamond^{2*}a$, and clearly $L_0L$ is a subfunctor of the algebraically-free monad $T$ on $L_0 + L$. 

\begin{proof}
We know that if $T$ is the algebraically free monad on $L_0+L$, then $T\A$ is a free $L_0+L$-algebra on $\A$. Hence, to obtain $\unf^*$, we need to exhibit $\A + L_0T\A$ as an algebra
\[
k:\A + L_0(\A + L_0T\A) + L(\A + L_0T\A)\to \A + L_0T\A
\]
The first component can simply be included; on the second component, we have a natural map $\A + L_0T\A\overset{\inl}{\to} \A + L_0T\A + LT\A\to T\A$ using the free $(L_0+L)$-algebra structure on $T\A$; applying $L_0$ to this map yields a map $L_0(\A + L_0T\A) \to L_0T\A$; finally, for the third component, we get
\[
\begin{tikzcd}
L(\A + L_0T\A)\arrow[r, "\unf"] &
\A + L_0T\A + L_0T(\A + L_0T\A) 
 \arrow[r, "\text{algebra}"] &
\A + L_0T\A + L_0TT\A \\
\arrow[r, "\text{monad}"] &
\A + L_0T\A + L_0T\A\arrow[r, "\text{codiagonal}"] & 
\A + L_0T\A
\end{tikzcd}
\]
The resulting algebra map $\unf^*:T\A \to \A + L_0T\A$ is natural in $\A$ as it is induced by the composition
of natural transformations.
\end{proof}

\section{Initial Algebra Semantics}\label{sec:initsem}

The definition of the semantic map above is somewhat different from the usual definition in terms of initial algebra semantics. 
The reason for this is that a coalgebra map $\gamma:X\to BX$ does not 
directly give rise to an $L$-algebra structure on $PX$. With some additional work, 
we can exhibit the semantic map in this way as well. This section is devoted to 
presenting an initial algebra semantics, and proving it equivalent to the least-solution semantics.

Before we start, however, we will give our reasons for choosing the above construction as fundamental. Firstly, the definition of $\lbk-\rbk$ as a `least solution' corresponds more closely to intuitions about the semantics of (least) fixpoint formulas. Secondly, it yields a direct proof technique via the ordinal approximation sequence. The proof of invariance under behavioral equivalence from section \ref{sec:invariance} would be more involved if one were to work with the initial algebra semantics. 

For the main result of this section, the proof of equivalence between initial and least-solution
semantics, we need the following very mild
categorical assumption.

\begin{assumption}\label{ass:precomp}
    We assume that in $\D$ pre-composition with morphisms distributes over arbitrary joins, i.e., for 
    $g: \A \to \A'$, an index set $I$ and $f_i: \A' \to \A''$ for all $i \in I$ we have
    \[ \bigvee_{i \in I} (f_i \circ g)=(\bigvee_{i \in I} f_i) \circ g \] 
    where the equality means that if the join on the right side of the equation exists, the join on the
    left side also exists and in this case both morphisms are equal. 
\end{assumption}
It can be easily verified that Assumption~\ref{ass:precomp} is always 
satisfied in case the poset structure on $\Hom_\D(\A,\A')$ is 
induced pointwise, as is usually the case.

For the definition of the initial algebra semantics we first recall that any $L$-algebra structure 
on an object $\B$ induces a morphism from the initial $L$-algebra.
\begin{definition}
 For an $L$-algebra $\beta:L \B \to \B$ we let
 $\ulcorner m\urcorner: \Phi \to \B$ denote the unique $L$-algebra morphism from the initial $L$-algebra
 $(\Phi,\alpha)$ to $(\B,\beta)$.
\end{definition}

We now define a suitable $L$-algebra structure on the algebra of predicates.
\begin{definition}
Let $\gamma:X\to BX$ be a coalgebra. We define $m_\gamma:LPX\to PX$ as the least map $m$ making
\[
\begin{tikzcd}
LPX\arrow[rr, "m"]\arrow[dd, "\unf" '] && PX\\
&& PX + PBX\arrow[u, "{[PX,P\gamma]}" ']\\
PX + L_0LPX \arrow[rr, "PX + L_0m"] && PX + L_0PX\arrow[u, "PX+\delta" ']
\end{tikzcd}
\]
commute.
We define $\|-\|:\Phi\to PX$ as $\ulcorner m_\gamma \urcorner$, i.e., as
the unique $L$-algebra morphism; that is, the unique map making the following diagram commute
\[
\begin{tikzcd}
\Phi\arrow[r, "\|-\|"] &PX\\
L\Phi\arrow[u, "\alpha"]\arrow[r, "L\|-\|"] & LPX\arrow[u, "m_\gamma"]
\end{tikzcd}
\]
\end{definition}

Similarly as for the semantic map $\smap$ it is easy to see that $m_\gamma$ - and thus $\|-\|$ -  is well-defined as $m_\gamma$ is the least fixpoint of the {\em monotone} operator
$\Opinit$ given by $\Opinit (m) \mathrel{:=} [PX,P\gamma \circ \delta] \circ (PX
+ L_0 m) \circ \unf$.
We wish to prove that $\|-\| = \lbk-\rbk$. We first consider the easier direction of the statement.

\begin{lemma}\label{lem:sem1}
   We have $\|-\| \geq \lbk-\rbk$.
\end{lemma}
\begin{proof}[Sketch]
The lemma can be proven by showing that $\|-\|$ makes the definition square of $\smap$ commute.
As $\smap$ is the least such map, the claim follows. 
\end{proof}

In order to prove the other inequality, we first prove the following lemma.

\begin{lemma}\label{lem:sem2}
    Let $\{\Opinit^i\}_{i \in \Ord}$ denote the approximants of $m_\gamma$ in $\Hom_\D(LPX,PX)$. We have that
    $\fold{\Opinit^i} \leq \smap$ for all ordinals $i \in \Ord$.
\end{lemma}
\begin{proof}
We first prove that for all ordinals $i$ we have
\begin{equation}\label{equ:lax}
\Opinit^i \circ L \smap \circ \alpha^- \leq \smap
\end{equation}
In other words, we show that
$\smap$ is a pre-fixpoint of the following monotone operator on
$\Hom_\D(\Phi, PX)$ given by$
g\mapsto \Opinit^i \circ Lg\circ \alpha^{-1}$.
The claim of the lemma follows that by the fact that
$\fold{\Opinit^i}$ is the unique fixpoint of the above operator and thus also the smallest pre-fixpoint.
The proof of~(\ref{equ:lax}) is by ordinal induction on $i$. For $i=0$ the claim is trivial. For the successor ordinal case 
we calculate:
\begin{eqnarray*}
    \smap &\stackrel{\mbox{\tiny Def.}}{=}& 
    [PX,P\gamma \circ \delta] \circ (\smap+L_0L \smap) \circ (\Phi + L_0 \alpha) 
    \circ \unf \circ \alpha^- \\
         & \stackrel{\mbox{\tiny (*)}}{\geq} & 
         [PX,P\gamma \circ \delta] \circ (PX+L_0 \Opinit^i) \circ (\smap+L_0L \smap) 
         \circ \unf \circ \alpha^- \\
         & \stackrel{\mbox{\tiny nat of $\unf$}}{=} &
         [PX,P\gamma \circ \delta] \circ (PX+L_0 \Opinit^i) \circ \unf \circ L \smap \circ \alpha^- \\
         & \stackrel{\mbox{\tiny Def of $\Opinit^{i+1}$}}{=} & \Opinit^{i+1} \circ L \smap \circ \alpha^-
\end{eqnarray*}
Here (*) is a consequence of the I.H. and the fact that composition in $\D$ is monotonic. 
For the limit case consider $\Opinit^i = \bigvee_{j<i} \Opinit^j$. We have:
\[ \smap \stackrel{\mbox{\tiny I.H.}}{\geq} \bigvee_{j<i} \left(\Opinit^{j} 
\circ L \smap \circ \alpha^-\right)  \stackrel{\mbox{\tiny Ass.}}{\geq} \left(\bigvee_{j<i}\Opinit^{j} \right) \circ L \smap \circ \alpha^- \]
where for the second inequality we made use of Assumption~\ref{ass:precomp} that pre-composition of morphisms in $\D$ distributes
over joins.
\end{proof}

We are now ready to prove the main result of this section.
\begin{theorem}
    Let $(X,\gamma)$ be a $B$-coalgebra.
    and let $(\unf,\delta)$ be an unfolding system for
    $B$.
    The semantic
    map $\smap_\gamma: \Phi \to PX$ coincides 
    with the initial algebra map 
    $\initsem{-}_X$
    from $(\Phi,\alpha)$ to $(PX,m_\gamma)$
\end{theorem}
\begin{proof}
    By Lemma~\ref{lem:sem1} we have $\smap_\gamma \leq \initsem{-}_X$.
    For the converse direction note that we have
    $m_\gamma = \Opinit^i$ for some $i \in \Ord$.
    Therefore we can apply Lemma~\ref{lem:sem2} as follows:
    $\initsem{-}_{PX} = \fold{m_\gamma} = \fold{\Opinit^i}
    \stackrel{\mbox{\tiny Lem.~\ref{lem:sem2}}}{\leq} 
    \smap_\gamma$
\end{proof}

\section{Correctness for Positive Alternation-Free Coalgebraic Fixpoint Logics}\label{sec:afcfl}

In this section, we show how our approach can express the alternation-free fragment of the coalgebraic $\mu$-calculus as defined in \cite{cirstea2011tableaux}. More specifically, we give a pair of logical functors $(L,L_0)$ and an unfolding system $(\unf,\delta)$ such that the initial $L$-algebra $\Phi$ contains `enough' formulas of the $\mu$-calculus, and the induced semantics $\smap$ coincides with the ususal semantics for the $\mu$-calculus. 

Before we can give the unfolding system, we need to define a special syntax for the $\mu$-calculus. The reason for this is that our setup is most suited for logics with fixpoint \emph{modalities}; however, in the $\mu$-calculus, there is no syntactical distinction between formulas, and fixpoint modalities applied to a formula. 

In order to get the $\mu$-calculus into a more amenable form, we use a presentation similar to the flat coalgebraic fixpoint logics from \cite{schroder2010flat}. This approach can also be compared to PDL, which similarly has a syntactic distinction between \emph{formulas} and \emph{programs}, corresponding to our fixpoint schemes. 

Fix a countable set $V$ of `parametric variables'. Let $\Lambda$ be modal similarity type, i.e. a countable set of modal operators of finite arity. Let $x\notin V$ be a designated `fixpoint variable'. Then we define by mutual induction the formulas $\phi$ and fixpoint schemes $\gamma$ as follows:
\begin{align*}
\gamma &::= v\in V\mid x\mid \phi\mid \gamma\vee\gamma\mid\gamma\wedge\gamma\mid \heart(\bar\gamma)\mid \sharp_{\gamma}(\bar\gamma/\bar v)\mid \flat_{\gamma}(\bar\gamma/\bar v)\\
\phi &::=\top\mid\bot\mid \phi\vee\phi\mid\phi\wedge\phi\mid \neg\phi \mid \heart(\bar\phi)\mid \sharp_{\gamma}({\bar\phi}/{\bar v})\mid \flat_\gamma(\bar\phi/\bar v)
\end{align*}
where the notation $\bar\gamma,\bar \phi$ indicates that operators might
have lists of fixpoint schemes/formulas as arguments. The intended reading of $\sharp_\gamma$ and $\flat_\gamma$ is as taking the least and greatest fixpoints respectively of the operator defined by $\gamma$. We also define the set of \emph{guarded} fixpoint schemes via
\[
\gamma_g ::= v\in V\mid \gamma_g\vee\gamma_g\mid \gamma_g\wedge\gamma_g\mid \heartsuit(\bar\gamma)
\]
That is, a guarded fixpoint scheme is a propositional combination of parametric variables and one-step modalities applied to (non-guarded) fixpoint schemes.

 Note that negations may only occur in formulas, not fixpoint schemes; hence, the fixpoint variable $x$ only ever occurs positively. Note also that on formulas, the `greatest fixpoint operator' $\flat_\gamma$ may be defined in terms of $\sharp_\gamma$ via
$
\flat_\gamma(\bar \phi/\bar v) := \neg \sharp_\gamma(\overline{\neg\phi}/{\bar v})
$.
We write CFL for the above coalgebraic fixpoint logic and CFS for the set of fixpoint schemes. We will also write CFS${}_\sharp$ for the set of fixpoint schemes not containing $\flat$, CFL${}_\sharp$ for the set of formulas not containing $\flat$, and CFL${}_\sharp^+$ for the set of formulas not containing $\flat$ or $\neg$. 

\begin{definition}\label{def:musicalsemantics}
Let $B$ be a $\Sets$-functor, and fix for every
 modality $\heart$ of arity $n$ a monotone predicate lifting 
 $\lds \heart\rds:P^n\to PB$. Then for a $B$-coalgebra 
 $\sigma:X\mapsto BX$, we define by mutual recursion the 
 semantics $\lbk \phi\rbk$ of a formula, 
 and $\lbk\gamma\rbk_p^\xi$ of a fixpoint scheme, where $p:V\to PX$ 
 is a valuation of the parametric variables, 
 and $\xi\subseteq PX$ is the value of the fixpoint variable
\begin{align*}
&\lbk \phi\rbk_p^\xi := \lbk \phi\rbk&\\
&\lbk v\rbk_p^\xi := p(v)& \lbk \top\rbk := S\\
&\lbk x\rbk_p^\xi := \xi &\lbk \bot\rbk := \varnothing\\
&\lbk\gamma\vee\delta\rbk_p^\xi := \lbk \gamma\rbk_p^\xi \cup \lbk\delta\rbk_p^\xi&\lbk \phi \vee\psi\rbk := \lbk \phi\rbk\cup \lbk \psi\rbk\\
&\lbk\gamma\wedge\delta\rbk_p^\xi := \lbk \gamma\rbk_p^\xi \cap \lbk\delta\rbk_p^\xi&\lbk \phi \wedge\psi\rbk := \lbk \phi\rbk\cap \lbk \psi\rbk\\
&\lbk \heart(\bar\gamma)\rbk_p^\xi := \sigma^{-1}\left(\lds\heart\rds(\lbk \bar\gamma\rbk_p^\xi)\right)&\lbk \heart(\bar\phi)\rbk := \sigma^{-1}\left(\lds\heart\rds(\lbk \bar\phi\rbk)\right)\\
&\lbk\sharp_\gamma(\bar\delta/\bar v)\rbk_p^\xi:= \mu\left(U\mapsto \lbk \gamma\rbk_{\lbk \bar\delta\rbk_p^\xi/\bar v}^U\right)&\lbk \sharp_\gamma(\bar \phi/\bar v)\rbk:= \mu\left(U \mapsto \lbk\gamma\rbk_{\lbk \bar\phi\rbk/\bar v}^U\right)\\
&\lbk\flat_\gamma(\bar\delta/\bar v)\rbk_p^\xi:= \nu\left(U\mapsto \lbk \gamma\rbk_{\lbk \bar\delta\rbk_p^\xi/\bar v}^U\right)&\lbk \flat_\gamma(\bar \phi/\bar v)\rbk:= \nu\left(U \mapsto \lbk\gamma\rbk_{\lbk \bar\phi\rbk/\bar v}^U\right)
\end{align*}
\end{definition}

We claim that CFL is equivalent to the coalgebraic $\mu$-calculus (CMC), and CFL${}_\sharp$ is equivalent to the alternation-free fragment of the coalgebraic $\mu$-calculus (AFCMC). 
We will illustrate by showing how to express two $\mu$-calculus formulas in CFL. 

\begin{example}\label{ex:mucalc1}
As a simple example, consider the PDL formula $\Diamond^*p$. In the $\mu$-calculus, this formula may be expressed as
$
\phi := \mu X. p \vee \Diamond X
$.
We can express $\phi$ in CFL as follows:
$
\sharp_\gamma(p/v)$, where $\gamma(v;x) := v \vee \Diamond x
$.
Note that this is not exactly the formula generated by the translation; in fact we have
\[
\trl(\phi) = \sharp_\delta(),\qquad\text{ where }\delta(x) = p \vee \Diamond x
\]
We have chosen $\sharp_\gamma(p/v)$ as our intuitive translation, since it highlights how the fixpoint \emph{modality} $\Diamond^*$ may be expressed as the modality $\sharp_\gamma$. 
\end{example}

\begin{example}\label{ex:mucalc2}
As a second example, we look at the formula
\[
\phi = \mu X . p \wedge \mu Y . \left((q\wedge \Diamond Y)\vee (r\wedge \square X)\right).
\]
This example exhibits intrinsic nesting of fixpoints. It is guarded as a formula of the $\mu$-calculus; however, in CFL, guardedness also requires that all fixpoint \emph{quantifiers} (other than the first) appear directly under a modality. So, let $\psi$ be the formula bound by $Y$ - i.e., $\psi = (q\wedge \Diamond Y)\vee (r\wedge \square X)$. Clearly, $\phi$ is equivalent to
\[
\mu X . \big(p\wedge \left((q\wedge \Diamond (\mu Y. \psi))\vee (r\wedge \square X)\right)\big)
\]
and in this formula, all but the outer quantifier appear guarded. Now translation is a straightforward affair: consider
\begin{align*}
\sharp_\gamma() & \text{ where }\\
& \gamma(x) := p\wedge \left((q\wedge \Diamond \sharp_\delta(x))\vee (r\wedge \square x)\right)\\
& \delta(v;x) := (q\wedge \Diamond x)\vee (r\wedge \square v)
\end{align*}
Then $\binom{\sharp_\gamma()}{\sharp_\delta(\sharp_\gamma()/v)}$ is the (mutual) least fixed point of
\[
\binom{X}{Y}\mapsto \binom{p\wedge ((q\wedge \Diamond Y)\vee (r\wedge \square X))}{(q\wedge \Diamond Y)\vee (r\wedge \square X)}
\]
which shows that $\sharp_\gamma()$ has the same interpretation as $\phi$. 
\end{example}

Next, we show how CFL${}_\sharp^+$ fits into the dual-adjunction picture. 
We take $\C$ to be the category $\Sets$, and $\D$ to be the category $\DLatt$ of distributive lattices. 
$\DLatt$ is enriched over $\Poset$ by the pointwise order. 
For the adjunction $P \dashv Q$ we take the adjunction from Example~\ref{ex:positivemodal} in Section~\ref{sec:onestep}.

Our behavior functor $B:\Sets\to\Sets$ was given, together with a set of predicate liftings $\{\lds \heartsuit \rds\mid \heartsuit\in\Lambda\}$. Similarly to the one-step logic from Example~\ref{ex:positivemodal}, we take $L_0$ to be given by
\[
L_0\A := \Free(\{\heart(a_1,\dots, a_n)\mid a_i\in A, \heart \in \Lambda\text{ of arity }n\})/{\approx}
\]
where $\approx$ is the least congruence such that if 
$a\leq b$, then $\heartsuit a \leq \heartsuit b$. 
Now $\lds-\rds$ can be used to define a one-step semantics $\delta: L_0P\to PB$ (cf.~\cite[Ex.~3.11]{KupkeP11}). 
Next, we set
\begin{eqnarray*}
L\A & := & \Free(\{\sharp_\gamma(\bar a/\bar v)\mid \gamma\in \text{CFS}_\sharp^\mu, \gamma\text{ guarded}\})/{\approx}
\end{eqnarray*}
Note that $L_0 + L$ generates an algebraically-free monad $T$, the monad of \emph{$L_0 + L$-terms}. Since each $\gamma$ is an $L_0 + L$-term in free variables $\bar v\cup x$, we obtain a natural transformation $\unf:L\to T$ via
$
\unf:\sharp_\gamma(\bar a/\bar v)\mapsto \gamma(\bar a/\bar v, \sharp_\gamma(\bar a/\bar v)/x)
$.
Since we stipulate that $\gamma$ is guarded, we know that $\unf$ will in fact factor through $\id + L_0T$. Hence, using proposition \ref{prop:depth}, we are able to define the desired semantic maps $\lbk-\rbk$. 

\begin{proposition}\label{prop:musicalcorrect}
Let $\sigma:X\to BX$ be a $B$-coalgebra, and let $\lbk-\rbk:\Phi\to PX$ be the semantic map induced by the unfolding system $(\delta, \unf)$. Then $\lbk\phi\rbk$ is the subset of $X$ defined as in definition \ref{def:musicalsemantics}. 
\end{proposition}
\begin{proof}[Sketch]
The argument for this is similar to the one given in example \ref{ex:diamondstar}: let $t:\Phi\to PX$ be the explicitly defined semantics. Then $t$ is clearly a solution to the diagram from definition \ref{def:smap}, and hence $t\geq \smap$. For the other direction, we can set up an approximation process for $t$, and prove by induction that each approximant is contained in all solutions. 
\end{proof}

Now since we have shown how to faithfully translate the $\mu$-calculus into CFL, this 
allows us to interpret the negation-free, $\nu$-free fragment of the 
coalgebraic $\mu$-calculus. In section \ref{sec:negations}, we will 
sketch how to extend this to the full alternation-free fragment of the coalgebraic $\mu$-calculus.

\section{Negations}\label{sec:negations}

So far, we have only been working in the positive fragments of the modal logics under consideration. However, many modal (fixpoint) logics come equipped with operations that are not order preserving - think of negation $\neg$ in Boolean algebras, or subtraction in quantitative algebras. We will use the term `negation(s)' as a stand-in for such operations generally. In this section, we sketch how the semantics on a positive fragment may be extended to the full logic with negations.

%

Let $\D^\neg$ be a category of `algebras-with-negations', equipped with a free-forgetful adjunction
$F:\D \to \D^\neg$, $U:\D^\neg \to \D$ with $F \dashv U$. 
This yields a endofunctors $L_0^\neg, L^\neg:\D^\neg\to \D^\neg$ given by $L^\neg = FLU, L_0^\neg = FL_0U$. We also assume that we have an adjunction $P^\neg \dashv Q^\neg$ with $P^\neg:\C\to(\D^\neg)^\op$,
and moreover assume that $UP^\neg = P$. Using this, the map $m$ from Section~\ref{sec:initsem} has type $m:LUP^\neg X\to UP^\neg X$. Hence, we can transpose it along $F\dashv U$ to get
\[
m^\neg:FLUP^\neg = L^\neg P^\neg X\to P^\neg X
\]
So, if we write $\Phi^\neg$ for the initial $L^\neg$-algebra, we obtain an initial algebra semantics for formulas with negations, which respects the already defined semantics for negation-free formulas. 

As an example, consider the logic CFL from Section~\ref{sec:afcfl}. 
As written there, we can only interpret CFL${}_\sharp^+$, corresponding to the fragment of CMC without greatest fixpoints and negations. However, if we take $\D^\neg$ to be $\BA$, 
the category of Boolean algebras, we see that the adjunction $P_\DLatt \dashv Q$ factors as 
$U P_\BA \dashv \uf F$, where $U:\BA\to\DLatt$ is forgetful, $F:\DLatt\to \BA$ is free,
$P_\DLatt$ and $P_\BA$ are the powerset functors into $\DLatt$ and $\BA$, respectively,  
and $P_\BA \dashv \uf$ is the well-known powerset-ultrafilter adjunction. 

By the above discussion, we also get an interpretation of $L^\neg$-formulas in $P_\BA X$ for any coalgebra $\xi:X\to BX$. The concrete analogue of this shift consists of adding in negations to the syntax of formulas (but not fixpoint schemes)
\begin{align*}
\gamma &::= v\in V\mid x\mid \phi\mid \gamma\vee\gamma\mid\gamma\wedge\gamma\mid \heart(\bar\gamma)\mid \sharp_{\gamma}(\bar\gamma/\bar v)\\
\phi &::=\top\mid\bot\mid \phi\vee\phi\mid\phi\wedge\phi\mid \heart(\bar\phi)\mid \sharp_{\gamma}({\bar\phi}/{\bar v})\mid \boldsymbol{\neg \phi}
\end{align*}
yielding CFL${}_\sharp$, which corresponds to the full alternation-free fragment of the coalgebraic $\mu$-calculus.

The picture is similar to that in \cite{balan2014positive}. There, the authors define `positive fragments' of one-step logics (however restricted to the logical connection between $\Sets$ and $\BA$). For our purposes, we would prefer to start with a (positive) logic, and consider `negative extensions'. There is also the interesting question of expanding the notion of a `positive fragment'/ `negative extension' beyond the pair $\BA/\DLatt$, to more quantitative logics. 

It is important to note that, while this section shows how to extend the initial algebra semantics of the positive logic to a logic with negation, it is not clear whether we can define the semantics for the logic with negation directly, i.e., based on the unforlding system as 
in~Definition~\ref{def:smap}.

\section{Conclusion}\label{sec:conclusion}

We have provided a new categorical framework for
studying coalgebraic fixpoint logics based on the dual adjunction framework.
The framework provides both a least-solution and an initial algebra semantics. 
We exemplified the framework using a number of different examples such as 
the positive modal logic of transitive closure, its probabilistic variant, positive PDL 
and the positive alternation-free fragment of the coalgebraic $\mu$-calculus. We also
showed how to add negations to these logics, but hasten to remark that in this case
we only have an initial algebra semantics.
As a first simple example of how the framework can be used we provided a generic proof of adequacy 
of our fixpoint logics, i.e., the semantics of any logic that fits in the framework is invariant
under behavioural equivalence. 

   Our framework is based on adding recursion to coalgebraic modal logic where the one-step semantics is given by a type of distributive law. In~\cite{DBLP:journals/jlp/RotB16}, instead, it is shown how to add a form of recursion to \emph{abstract GSOS specifications}---a different type of distributive laws that can be used to study structural operational semantics at a high level of generality---using ``unfolding'' natural transformations, similar to the current approach. An interesting theoretical direction may be to try and unify these extensions; a good starting point may be the steps-and-adjunctions framework studied in~\cite{rot2021steps}.

The main contribution of this paper is the new 
categorical framework for fixpoint logics. In the future we hope to exploit
this in several ways: Firstly, we will study proof systems~\cite{focusMartiV21} 
and reasoning procedures~\cite{HausmannSE16} for alternation-free 
fixpoint logics and will lift soundness and completeness proofs to our
framework. Secondly, the alternation-free fragment of the modal $\mu$-calculus has interesting
model-theoretic properties~\cite{facchinivenema2013} and we will explore whether those can be
recovered within our framework. Another important question -- suggested by two of the anonymous
referees -- is whether we can use our framework to provide a completeness proof for 
the Segerberg axioms for PDL~\cite{kozepari1981} or their coalgebraic generalisation from~\cite{hanskupk2015}. These logics certainly do fit into our framework, the key question to solve will be 
to formally connect the induction axiom or the least fixpoint rule to the requirement that 
the semantic map is defined as a least fixpoint.  
Finally, the standard dual adjunction framework 
can be used to define certain filtrations~\cite{BarloccoKR19}. We have  reasons to believe
that our framework of unfolding systems can be used to show that a (fixpoint) logic has filtrations.
This opens an avenue for studying filtrations and abstraction techniques for fixpoint logics 
in a categorical way.

\section*{Acknowledgments}

The authors would like to thank the anonymous referees for instructive comments and interesting questions for 
future work, as well as for suggesting the formula in example \ref{ex:mucalc2}.

\bibliographystyle{plain}
\bibliography{sandi.bib}

\end{document}